\documentclass[aps,prl,reprint,twocolumn,superscriptaddress,groupedaddress]{revtex4}  

\usepackage{graphicx}  
\usepackage{dcolumn}   
\usepackage{bm}        
\usepackage{amssymb,amsmath,amsthm}   
\usepackage{thmtools,thm-restate} 
\usepackage[utf8]{inputenc}
\usepackage[english]{babel}
\usepackage[T1]{fontenc}
\usepackage{hyperref}
\usepackage{tikz}
\usetikzlibrary{shapes,snakes,calc}
\usetikzlibrary{backgrounds,fit,decorations.pathreplacing}  
\usepackage{xcolor}
\usepackage[inline]{enumitem}
\usepackage[normalem]{ulem}
\usepackage{nicefrac}
\usepackage{xfrac}
\usepackage{bm}
\usepackage{bbm}
\usepackage{lipsum}
\usepackage{qcircuit}

\newcommand{\ket}[1]{|#1\rangle}

\theoremstyle{plain}
\newtheorem{theorem}                 {Theorem}
\newtheorem{proposition}             {Proposition}

\theoremstyle{remark}

\newcommand{\prlsection}[1]{\paragraph*{#1.---}\hspace{-1em}}

\begin{document}

\let\reusablemaketitle\maketitle
\let\reusableauthor\author
\let\reusableaffiliation\affiliation

\newcommand{\affiltwo}{School of Informatics, University of Edinburgh, 10 Crichton St, Edinburgh EH8 9AB, United Kingdom}
\newcommand{\affilone}{Sorbonne Université, CNRS, Laboratoire d'Informatique de Paris 6, F-75005 Paris, France}

\title{Quantum Advantage from Sequential-Transformation Contextuality}

\author{Shane Mansfield}%
\email{shane.mansfield@lip6.fr}%
\affiliation{\affilone}%

\author{Elham Kashefi}%
\affiliation{\affilone}%
\affiliation{\affiltwo}

\date{\today}

\begin{abstract}
We introduce a notion of contextuality for transformations in sequential contexts, distinct from the Bell-Kochen-Specker and Spekkens notions of contextuality.
Within a transformation-based model for quantum computation we show that strong sequential-transformation contextuality is necessary and sufficient for deterministic computation of non-linear functions if classical components are restricted to $\bmod{2}$-linearity and matching constraints apply to any underlying ontology. For probabilistic computation, sequential-transformation contextuality is necessary and sufficient for advantage in this task and the degree of advantage quantifiably relates to the degree of contextuality.
\end{abstract}


\maketitle

Contextuality is a key non-classical phenomenon exhibited by quantum systems, which was first considered by Bell \cite{bell2} and by Kochen and Specker \cite{ks}. It has been the subject of renewed interest recently as a range of results have established it to be the essential ingredient for enabling  quantum advantages over classical implementations of a variety of informatic tasks \cite{andersbrowne,raussendorfcontextuality,abmprl}, simulation of quantum processes \cite{karanjai2018contextuality}, and for enabling universal quantum computing \cite{howardmagic,delfosse2015wigner,pashayan2015estimating,juanito2017,catani2017state}\footnote{Some of these references build on earlier work relating quantum advantage to Wigner function negativity \cite{galvao2005discrete,veitch2012negative}, known to be an equivalent notion of non-classicality to contextuality \cite{spekkens2008negativity}.}. A broader notion of contextuality due to Spekkens \cite{spekkens2009preparation} has also been shown to be essential to quantum advantages relating to state discrimination and one-way communication protocols \cite{chailloux2016optimal,saha2017state,schmid2018contextual,saha2018preparation,ghorai2018optimal}. However, questions remain over which forms of contextuality provide advantage in which precise settings \cite{lillystone2018contextuality} and whether existing notions of contextuality are sufficient to account for all instances of quantum advantage.
For example, there exist a variety of advantages achievable with a single qubit \cite{knill1998power,galvao2003substituting,dkk}, where Bell-Kochen-Specker (BKS) contextuality cannot arise \cite{ks,gleason1957measures} and to which there is no apparent link to the Spekkens version. This raises the important question of which non-classical feature could be at play if not contextuality of these kinds.

We introduce a notion of contextuality for transformations performed in sequential contexts that is inequivalent to the notion of transformation contextuality introduced by Spekkens \cite{spekkenscontextuality}. It is necessarily present in a recently discovered form of quantum advantage in shallow circuits \cite{bravyi2017quantum}. We will show, via a Mermin-style \cite{mermin:90a,mermin:90} parity argument, that it is also crucial in enabling increased computational power in the single qubit example of \cite{dkk}.
The setting for that example is a transformation-based model of quantum computing, which we call here $l2$-TBQC, that was shown to be useful in achieving secure delegated computing. In the model, a classical control computer, whose power is limited to $\bmod{2}$-linear computation, may interact with a quantum resource, by which its computational power may be enhanced. As with the analogous measurement-based model, $l2$-MBQC \cite{raussendorfcontextuality}, which was the setting for the results of \cite{andersbrowne,raussendorfcontextuality,abmprl},
it can provide a useful tool for probing the roots of quantum advantage. In this setting, we show more generally that sequential-transformation contextuality is necessary and sufficient to enable advantage in the task of probabilistically computing any non-linear function whenever classical ontologies are required to respect the computational assumptions. Moreover, the degree of contextuality can be related to the probability of success, and in particular, strong (i.e.,~maximal) contextuality is necessary for deterministic computation of any non-linear function.

Our results trace an arc that parallels developments relating BKS contextuality to quantum advantage in $l2$-MBQC: Anders and Browne provided an example in which a contextual resource is sufficient for the computation of a particular non-linear function \cite{andersbrowne}; Raussendorf then proved that strong contextuality is necessary for any deterministic non-linear computation \cite{raussendorfcontextuality}, as initially observed by Hoban \textit{et al.}~for non-adaptive $l2$-MBQC \cite{hoban2011non} based on an early version of \cite{raussendorfcontextuality};
he also showed that contextuality is necessary for quantum advantage in the task of probabilistically computing any non-linear function; this latter result was later sharpened to show more precisely how the degree of contextuality as measured by the contextual fraction relates to probability of success in \cite{abmprl}. Our results set the stage for further investigation of how sequential-transformation contextuality may relate to quantum advantages, speedups and the onset of universality in other settings, as the results of \cite{howardmagic,delfosse2015wigner,pashayan2015estimating,juanito2017} do for BKS contextuality.

\prlsection{Ontological models}
Quantum theory exhibits a number of apparently non-intuitive features. Crucially, in many cases there exist no-go theorems that establish that there is no way these features can be explained away by recourse to any deeper or more complete theory that would obey certain classical intuitions \cite{epr}. Some such non-classical features are non-locality \cite{bell1} (BKS) contextuality \cite{bell2,ks} \footnote{There now exist a number of unified treatments of non-locality and (BKS) contextuality \cite{abramskybrandenburger,afls,csw,cbd}, which are closely interrelated \cite{afls,horse1}, and whose relations to the other kinds of non-cassicality have also been explored elsewhere \cite{cbdmacro,wester,kunjwal2017beyond}.}, forms of preparation and transformation contextuality \cite{spekkenscontextuality}, while others relate to macro-realism \cite{lg,entanglementintime,timpsonmaroney,JMAmacro} and the ontic nature of the quantum state \cite{liang:11,pusey:12,colbeck:12,colbeck:13,hardy:13a,montina:15,mansfield:16,JMAontic}.
A convenient formalism for treating such theorems is that of ontological models, which we briefly set out next. Note that in this work when we speak of ontological models we will not be assuming any additional features beyond what is explicitly set out below (e.g.,~of the kind present in \cite{spekkenscontextuality}).

The central component is an ontic state space $\Lambda$, comprising the states of a hypothetical underlying theory. Preparation of a quantum state $\rho$ results in an ontic state sampled according to a probability distribution $d_\rho$ on $\Lambda$ \footnote{Ontological theories can be defined more generally with measures, but this will not be necessary for our present purposes.}.
In the simplest case, a quantum transformation $U$ corresponds to a measurable function $f_U:\Lambda \rightarrow \Lambda$.
For consistency we require that $f_{U*}d_\rho=d_{U \rho U^\dagger}$, where the left-hand side is the push forward of $d_\rho$ along $f_U$, defined by $f_{U*}d_\rho(\lambda) = d_\rho [ f_{U}^{-1}(\lambda) ] $. We also require that the function corresponding to the identity operator simply maps each ontic state to the $\delta$ function centered on that state, ensuring that $f_{\mathbbm{1}*}d_\rho = d_\rho$ for all preparations $\rho$. 
In particular, the requirements entail that unitaries correspond to invertible functions.
A quantum measurement $M$ corresponds to a function $\xi_M : \Lambda \rightarrow P(O)$ which assigns to each ontic state a probability distribution over the set of outcomes $O$. For any combination of preparation, transformation, and measurement, the ontological theory predicts that the empirical statistics, $e_{\rho,U,M} \in P(O)$, are given by
\begin{equation}\label{eq:emp}
e_{\rho,U,M} = \sum_{\lambda \in \Lambda} \, d_\rho(\lambda) \, \xi_M(f_U(\lambda)) \, .
\end{equation}
In fact, our results apply more generally to ontological models in which transformations may correspond to stochastic mixtures of measurable functions. However, we will see shortly that, for our present purposes, since such an ontological model can always be expressed as a convex decomposition of ones in which transformations are deterministic, it will suffice to establish no-go properties for those with deterministic transformations.
No-go theorems arise when it is found that ontological models satisfying some additional, perhaps ``classical'', assumptions are unable to realise the empirical predictions of quantum theory.

\prlsection{(Non-)contextuality}
In the BKS sense, non-contextuality is an assumption of classicality that applies when certain finite sets of compatible measurements may be performed jointly in contexts. It requires that for each valid context $C$ compatibility is reflected at the ontological level through factorisability of the joint measurement function $\xi_C : \Lambda \rightarrow P(O^{\left| C \right|})$; i.e.,
\begin{equation}\label{eq:ks}
\xi_C  = \prod_{M \in C} \xi_{M} \, .
\end{equation}
Implicit in this is the crucial requirement that, for any measurement $M$ occurring in contexts $C$ and $C'$, its ontological representation $\xi_M$ is context independent; i.e.,
\[
\xi_{M^{(C)}} = \xi_{M^{(C')}} \, .
\]
This description of non-contextuality via factorisability is equivalent to the description in terms of global valuations that may be more familiar to some readers \cite{abramskybrandenburger}.


Next, we mention some specific instances arising from Spekkens' general notion of non-contextuality \cite{spekkenscontextuality}. Measurement non-contextuality in the explicit sense treated in the no-go results of \cite{spekkenscontextuality} relaxes (\ref{eq:ks}) to the weaker requirement that
\[
\left. \xi_C\right|_M = \xi_M \, ,
\]
for all $M$ and $C$ such that $M \in C$, where $\left. \xi_C\right|_M$ denotes the marginalisation of $\xi_C$ to $M$.

Transformation and preparation non-contextuality in the explicit sense treated in the no-go results of \cite{spekkenscontextuality} takes as context any convex decomposition of a given transformation or preparation. This has an operational motivation. Suppose, as a concrete example, that some transformation $T$ admits the following unitary decompositions:
\begin{align}
T &= \frac{1}{2} U_a + \frac{1}{2} U_A \label{eq:C}\tag{$C$} \, , \\
T &= \frac{1}{3} U_a + \frac{1}{3} U_b + \frac{1}{3} U_c \, . \label{eq:C'}\tag{$C'$}
\end{align}
Operationally, context \ref{eq:C} is ``apply $U_a$ or $U_A$ uniformly at random'', and context \ref{eq:C'} is ``apply $U_a$, $U_b$ or $U_c$ uniformly at random''; quantum mechanically the contexts are equivalent.
Non-contextuality requires that convex decompositions are reflected at the ontological level; i.e.,~in this instance,
\[
f_T = \frac{1}{2} f_{U_a} + \frac{1}{2} f_{U_A} = \frac{1}{3} f_{U_a} + \frac{1}{3} f_{U_b} + \frac{1}{3} f_{U_c}  \, .
\]
Again, it is implicit that ontological representations of transformations and preparations are independent of operational context; e.g.,
\[
f_{U_a^{(C)}} = f_{U_a^{(C')}} \, .
\]

\prlsection{Sequential transformations}
With the preceding versions for comparison, we now introduce a version of non-contextuality for transformations in sequential contexts. It requires that for each finite sequence of transformations, $C=\{U_i\}_{i=1}^t$, sequential composition is reflected at the ontological level; i.e.,
\begin{equation*}\label{eq:transnc}
f_{U_t\cdots U_1} = f_{U_t} \circ \dots \circ f_{U_1} \, .
\end{equation*}
It is assumed that the ontological representations of transformations are independent of sequential context; i.e.,~whenever a transformation $U$ occurs in contexts $C$ and $C'$, it holds that
\[
f_{U^{(C)}} = f_{U^{(C')}} \, .
\]
When a set of empirical data or predictions cannot be reproduced by an ontological model satisfying this property, it is said to be contextual.

Contextuality in our sense implies that the system of study cannot have an ontology in which transformations correspond to modular, composable operations on ontic states, such that they are well defined independently of which transformations may have been performed previously or will be performed subsequently. Either we must reject the ontological picture entirely or give up on these highly intuitive, classical properties. Note that one plausible, if conspiratorial, mechanism for introducing some contextuality might be through causal dependence on transformations having appeared earlier in the sequence, but even this kind of mechanism is precluded when the transformations being modelled commute.

The constant-depth quantum circuits of \cite{bravyi2017quantum} provide a concrete example of sequential-transformation contextuality as they can at best be simulated by classical circuits whose depth grows logarithmically in the size of the input. If a modular, non-contextual ontological description of gate transformations at each step in the circuit were possible, then it would give rise to classical circuits for the same task, which would also have constant depth. Connections to quantum advantage in this setting will be investigated in future work; here we focus on examples in a more restricted setting.

\prlsection{Quantification}
An empirical model $e = \{ e_C \}$, associates with each context $C$ a distribution over observed outcomes \cite{abramskybrandenburger}. Similar to \cite{abmprl}, given any empirical model and appropriate version of contextuality, we may consider convex decompositions of the form
\begin{equation}\label{eqn:edecomp}
e = \omega e^{\mathrm{NC}} + (1 - \omega) e' \, ,
\end{equation}
where $e^{\mathrm{NC}}$ and $e'$ are also empirical models, and $e^{\mathrm{NC}}$ is non-contextual. The maximum value of $\omega$ over all such decompositions is the non-contextual fraction of $e$, written $\mathrm{NCF}(e)$, and correspondingly, the contextual fraction of $e$ is $\mathrm{CF}(e) := 1 - \mathrm{NCF}(e)$ \footnote{BKS contextuality has also been quantified in such a manner in \cite{abramskybrandenburger,nccontent,grudka2014quantifying,thesis,abm-qpl16}.}.
For BKS contextuality, the contextual fraction corresponds to the maximum achievable normalised violation by $e$ of any generalised Bell inequality \cite{abmprl}. Here, however, we use it to quantify sequential-transformation contextuality. Using the terminology of the hierarchy of BKS contextuality introduced in \cite{abramskybrandenburger}, an empirical model is said to be strongly contextual when $\mathrm{CF}(e) = 1$.

For a given experimental scenario, the set of all the possible $e^{\mathrm{NC}}$ is convex, and any extremal point corresponds simply to fixing a deterministic function $f_U:\Lambda \rightarrow \Lambda$ for each transformation $U$ featuring in the scenario. Strong contextuality thus arises in the extreme case that no global assignment of deterministic functions to transformations is consistent with even a fraction of the empirical behaviour.

\prlsection{$l2$-TBQC}
We consider a classical control computer restricted to $\bmod{2}$-linear computation that can interact with a resource, which may be quantum, as follows. The resource is prepared in a fixed state, the control computer may interact with it by means of controlled transformations, then a fixed measurement is performed on the resource and its outcome returned to the control computer. This captures, for example, the single qubit protocols of \cite{dkk}, which were considered for their security features in a setting in which a client delegates certain operations making up an $l2$-TBQC, such as state preparation and measurement, to a server.

Note that, independent of the $l2$ restriction, any measurement-based quantum computation \cite{raussendorfbriegel} can equivalently be expressed as a TBQC, since choice of a measurement setting is equivalent to choice of a transformation prior to a fixed measurement.

An example of an $l2$-TBQC that performs a basic non-linear function, the $\mathrm{AND}$ gate on classical input bits $a$ and $b$,
\[
g(a,b) = (a \oplus 1) \otimes (b \oplus 1) \oplus 1 \, ,
\]
is the following (Fig.~\ref{fig})
\footnote{This is a slightly simplified version of the implementation from \cite{dkk}. Note that in the homomorphism from $\mathbb{Z}_2$ to the booleans which one would perform in order to interpret the function as the logical $\mathrm{AND}$ gate, the roles of $0$ and $1$ are exchanged, i.e.,~$0 \mapsto 1$ and $1 \mapsto 0$.}.
The control computer receives inputs $a$ and $b$. For the resource, the fixed state is the qubit state $\ket{+}$, the fixed measurement is given by the Pauli operator $\sigma_X$, and the controlled transformations are $U(a)$, then $V(b)$, then $W({a \oplus b})$, where
\begin{align*}
U(0) &= V(0) = W(0) = I \, ,\\
U(1) &= V(1) = W(1) = \left( \begin{array}{cc} 1 & 0 \\ 0 & e^{i \sfrac{\pi}{2}} \end{array} \right) \, .
\end{align*}
Notice that all transformations commute.
The output of the computation is the measurement outcome interpreted in $\mathbb{Z}_2$, with eigenvalues $+1$ and $-1$ mapped to $0$ and $1$, respectively.
In terms of complexity classes, access to a qubit quantum resource promotes the computational power from the class $\oplus L$ \cite{damm1990problems,aaronson2004improved} to $P$, as with the example in \cite{andersbrowne} in the setting of $l2$-MBQC.

\begin{figure}
\caption{\label{fig} The basic single qubit $\mathrm{AND}$ protocol from \cite{dkk}.}
\[
\Qcircuit @C=1em @R=1em {
\lstick{\ket{+}} & \gate{U(a)} & \gate{V(b)} & \gate{W(a\oplus b)} & \meter & \sigma_X
}
\]
\end{figure}
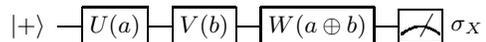

\prlsection{$\oplus L$-ontology}
Of course, classical computers can perfectly well compute non-linear functions and they also constitute valid non-contextual ontologies. To pose a meaningful computational question about whether a resource may be used to boost power from $\oplus L$ to $P$, therefore, we will restrict attention to $\oplus L$-ontologies, which we define as follows. Recalling that $\oplus L$ circuits are built entirely of $\mathrm{NOT}$ and controlled-$\mathrm{NOT}$ ($\mathrm{CNOT}$) gates \cite{aaronson2004improved}, we will suppose that available transformations are built from these and act on an ontic state space $\mathbb{Z}_2^s$, for some $s \in \mathbb{N}$.
In what follows, we will be interested in protocols in which transformations commute. These can already permit efficient solutions to problems for which it is believed there can be no efficient classical solution \cite{bremner2010classical}. For transformations in commutative $\oplus L$ ontologies it holds that, for any transformation $U$,
\begin{equation}\label{eq:ontf}
f_U(\bm{\lambda}) = \left( I \oplus A_U \right) \bm{\lambda} \oplus \bm{u} \, ,
\end{equation}
where $A_U$ is an $s \times s$ matrix over $\mathbb{Z}_2$ containing only off-diagonal entries and $\bm{u} \in \mathbb{Z}_2^s$ [see Appendix]. For composition of transformations $\{U_i\}_{i=1}^t$ with ontological representations determined by $\{A_i,\bm{u}_i\}$ it holds that
\begin{equation*}\label{eq:seq}
f_{U_t} \circ \cdots \circ f_{U_1}(\bm{\lambda}) = \bm{\lambda} \oplus \bigoplus_{i=1}^{t} A_{i} \bm{\lambda} \oplus \bigoplus_{i=1}^{t} \bm{u}_i \, .
\end{equation*}
A dichotomic measurement in an $\oplus L$ ontology can most generally be described by a transformation followed by output of the bit value of a fixed entry $j$ of the final ontic state vector; i.e.~$\bm{\lambda}' \cdot \bm{\delta}$ where $\bm{\lambda}',\bm{\delta} \in \mathbb{Z}_2^s$ are the post-transformation ontic state and the vector with $j$th entry $1$ and $0$'s elsewhere, respectively.

\section{Results}
\begin{proposition}\label{prop:dkk}
Any commutative $\oplus L$-ontological realisation of the $\mathsf{AND}$ $l2$-TBQC is transformation contextual.
\end{proposition}

\begin{proof}
Suppose that preparation results in an initial ontic state $\bm{\lambda} \in (\mathbb{Z}_2)^s$.
From (\ref{eq:emp}), non-contextual realisation of the protocol requires Eqs (\ref{c1}--\ref{c4}) to be satisfied. These describe evaluation of the computation for the four possible sequential contexts, where ontological representations of $U(k)$, $V(k)$, and $W(k)$, with $k\in\{0,1\}$, are determined through Eq.~(\ref{eq:ontf}) by $\{A_U(k),\bm{u}(k)\}$, $\{A_V(k),\bm{v}(k)\}$, and $\{A_W(k),\bm{w}(k)\}$, respectively, and of the transformation component of the measurement by $\{A_M,\bm{m}\}$,
\begin{widetext}
\begin{align}
\left[ \bm{\lambda} \oplus A_{U}(0) \bm{\lambda} \oplus A_{V}(0) \bm{\lambda} \oplus A_{W}(0) \bm{\lambda} \oplus A_M \bm{\lambda} \oplus \bm{u}(0) \oplus \bm{v}(0) \oplus \bm{w}(0) \oplus \bm{m} \right] \cdot \bm{\delta} = 0 \, ,
\label{c1}
\\
\left[ \bm{\lambda} \oplus A_{U}(0) \bm{\lambda} \oplus A_{V}(1) \bm{\lambda} \oplus A_{W}(1) \bm{\lambda} \oplus A_M \bm{\lambda} \oplus \bm{u}(0) \oplus \bm{v}(1) \oplus \bm{w}(1) \oplus \bm{m} \right] \cdot \bm{\delta} = 0 \, ,
\\
\left[ \bm{\lambda} \oplus A_{U}(1) \bm{\lambda} \oplus A_{V}(0) \bm{\lambda} \oplus A_{W}(1) \bm{\lambda} \oplus A_M \bm{\lambda} \oplus \bm{u}(1) \oplus \bm{v}(0) \oplus \bm{w}(1) \oplus \bm{m} \right] \cdot \bm{\delta} = 0\, ,
\\
\left[ \bm{\lambda} \oplus A_{U}(1) \bm{\lambda} \oplus A_{V}(1) \bm{\lambda} \oplus A_{W}(0) \bm{\lambda} \oplus A_M \bm{\lambda} \oplus \bm{u}(1) \oplus \bm{v}(1) \oplus \bm{w}(0) \oplus \bm{m} \right] \cdot \bm{\delta} = 1 \, .
\label{c4}
\end{align}
\end{widetext}
Under the assumption of non-contextuality, the equations are not jointly satisfiable. This can be deduced from the fact that the sum modulo $2$ of the right-hand sides is one, whereas the sum of the left-hand sides is zero, since each vector appears an even number of times leading to cancellations. Note that a contextual realisation would permit ontological representations to vary according to context; e.g.,~$\bm{u}(0)^{(4)} \neq \bm{u}(0)^{(5)}$. Contextually, we can always satisfy the equations. The conclusion is that, while $\oplus L$-ontological descriptions are possible, they are necessarily transformation contextual. 
\end{proof}

The above proof is similar to Mermin's parity version \cite{mermin:90a,mermin:90} of the Greenberger-Horne-Shimony-Zeilinger inquality-free argument for non-locality \cite{ghz:89,ghsz:90}, and is an instance of an all-versus-nothing proof of strong contextuality \cite{ccp}, albeit for transformation rather than BKS contextuality.

\begin{proposition}\label{prop:l2sc}
Strong transformation contextuality is necessary for $\oplus L$-ontological realisation of any non-linear commutative $l2$-TBQC.
\end{proposition}

\begin{proof}
Let $\bm{i} \in (\mathbb{Z}_2)^r$ be the input of the computation and $\bm{\lambda} \in (\mathbb{Z}_2)^s$ be the initial ontic state resulting from the fixed preparation. The control bits $\bm{k} = \{k_i\}_{i=1}^t$ for the transformations to be performed are linearly determined from the inputs:
\begin{equation}\label{eq:t}
\bm{k}=B\bm{i} \oplus \bm{c} \, .
\end{equation}
for some $n\times r$ matrix $B$ and vector $\bm{c}$ of length $r$ over $\mathbb{Z}_2$. The transformations to be performed in sequence are $\{ U_i (k_i) \}_{i =1}^t$. In a non-contextual model, these will have ontological representations determined through Eq.~(\ref{eq:ontrep}) by $\{A_i(k_i),\bm{u}_i(k_i)\}$. These are necessarily linear in $k_i$ for each $i$, since entries in $A_i$ and $\bm{u}_i$ take values in $\mathbb{Z}_2$ and are functionally determined from $k_i$, but all functions of type $\mathbb{Z}_2 \rightarrow \mathbb{Z}_2$ are linear. 
Deterministic realisation of a function $g: (\mathbb{Z}_2)^r \rightarrow \mathbb{Z}_2$ by a non-contextual ontological model requires that, for all inputs $\bm{i}$,
\begin{align*}
&g(\bm{i}) = \\
& \left[ \left( I \oplus \bigoplus_{i=1}^{t} A_{U_i}(k_i) \oplus A_M \right)  \bm{\lambda} \oplus \bigoplus_{i=1}^{n} \bm{u}_i(k_i) \oplus \bm{u}_M \right] \cdot \bm{\delta} \, .
\end{align*}
This is a linear function since right-hand side expression is linear in $\bm{k}$, and in turn $\bm{k}$ is linear in $\bm{i}$, by Eq.~(\ref{eq:t}).

In a contextual model we could allow the ontological representations to have context dependence; i.e.,~$\{A_i(\bm{k}),\bm{u}_i(\bm{k})\}$. In this case the entries of the respective matrices and vectors are determined by functions $\mathbb{Z}_2^n \rightarrow \mathbb{Z}_2$, which can introduce non-linearity.

Moreover, if any fraction $p$ of the empirical behaviour can be described non-contextually, then with an average probability over all possible inputs of at least $p$, the $l2$-TBQC computes some linear function. Therefore deterministic computation of a non-linear function requires strong contextuality.
\end{proof}

Given two functions $g,h: \mathbb{Z}_2^r \rightarrow \mathbb{Z}_2$, we can define an average distance between these functions as
\[
d(g,h) := 2^{-r}\left| \left\{ \bm{i} \mid g(\bm{i}) \neq h(\bm{i}) \right\} \right| \, .
\]
This can be used to measure the degree of non-linearity of any function $g: \mathbb{Z}_2^r \rightarrow \mathbb{Z}_2$ as the distance to the closest linear function of that type,
\[
{\nu}(g) := \min \left\{ {d}(g,h) \mid h: \mathbb{Z}_2^r \rightarrow \mathbb{Z}_2 \text{ linear} \right\} \, .
\]

\begin{theorem}\label{thm:main}
If a commutative $l2$-TBQC, with resource empirical model $e$, probabilistically computes a function $g: \mathbb{Z}_2^r \rightarrow \mathbb{Z}_2$ with an average failure probability ${\varepsilon}$ over all $2^r$ possible inputs, then
\[
{\varepsilon} \geq \mathsf{NCF}(e) \, {\nu}(g) \, .
\]
\end{theorem}

\begin{proof}
The average probability of success is ${p}_S:=1-{\varepsilon}$.
From (\ref{eqn:edecomp}), 
we can decompose the resource empirical model as
\[
e = \mathsf{NCF}(e) \, e^{\mathsf{NC}} + \mathsf{CF}(e) \, e' \, ,
\]
where $e'$ is necessarily strongly contextual. This allows us to similarly decompose the behaviour of the $l2$-TBQC, so that
\[
{p}_S = \mathsf{NCF}(e) \, {p}_{S,e^{\mathsf{NC}}} + \mathsf{CF}(e) \, {p}_{S,e'} \, ,
\]
where ${p}_{S,e^{\mathsf{NC}}}$ and ${p}_{S,e'}$ are the average probabilities of success that would be associated with resource empirical models $e^{\mathsf{NC}}$ and $e'$, respectively. At best, $e'$ enables deterministic computation of $g$. This leads to the inequalities
\begin{align}
{p}_S &\leq \mathsf{NCF}(e) \, {p}_{S,e^{\mathsf{NC}}} + \mathsf{CF}(e) \, , \nonumber \\
{\varepsilon} &\geq \mathsf{NCF}(e) \, {\varepsilon}_{e^{\mathsf{NC}}} \, , \label{eq:err}
\end{align}
where ${\varepsilon}_{e^{\mathsf{NC}}} = 1 - {p}_{S,e^{\mathsf{NC}}}$ is the average probability of failure associated with $e^{\mathsf{NC}}$.
From the proof of Propositon \ref{prop:l2sc}, 
we know that for a non-contextual resource empirical model any $l2$-TBQC can only compute convex mixtures of linear functions. Thus ${\varepsilon}_{e^{\mathsf{NC}}} \leq \tilde{\nu}(g)$, which combined with (\ref{eq:err}) yields the desired inequality.
\end{proof}

Theorem \ref{thm:main} extends Proposition \ref{prop:l2sc}, since, in particular it implies that deterministic computation (${\varepsilon}=0$) of a non-linear function [${\nu}(g)>0$] requires strong contextuality [$\mathsf{NCF}(e)=0$]. The proof here is
similar to that of Theorem 3 in \cite{abmprl}.

\section{Discussion}
The present results highlight the potential of sequential contextuality as a source of quantum advantage of a single qubit over arbitrarily many classical bits for a particular kind of computational task.
While the $\oplus L$-ontological assumptions are natural in the particular setting of restricted classical computation that we consider, a direction for future research will be to consider examples of sequential transformation contextuality in less restricted settings, like that of \cite{bravyi2017quantum}, as well as to explore other potential connections to quantum advantage, especially in single qubit systems \cite{knill1998power,galvao2003substituting}.
It also remains to be seen how the present notion of contextuality can be treated in resource-theoretic frameworks of the kind developed in \cite{grudka2014quantifying,horodecki2015axiomatic,abmprl,amaral2017noncontextual,duarte2017resource}.
A related analysis, in terms of irreversibility, of transformation-based protocols is contained in \cite{henaut2018tsirelson}, and, in the future, it may be interesting to consider advantages as arising from a combination of these phenomena.
From a foundational perspective, in light of the present analysis, the experimental results of \cite{barzenhanced,mod4} could already be said to provide indirect experimental evidence for a kind of sequential transformation contextuality, but this leaves open the possibility for experiments designed specifically to test for the feature, which might also aim to minimise potential issues, such as the detection loophole.

\prlsection{Acknowledgements}
The authors thank Samson Abramsky, Rui Soares Barbosa, Dan Browne, Ulysse Chabaud, Tom Douce, Pierre-Emmanuel Emeriau, Ernesto Galv\~{a}o, Frédéric Grosshans, Luciana Henaut, Matty Hoban, Aleks Kissinger, Jan-{\AA}ke Larsson, Damian Markham, and Anna Pappa for valuable discussions and comments.
Early ideas for this work were conceived while S.M.~was visiting the Simons Institute for the Theory of Computing at the University of California, Berkeley, as a participant of the Logical Structures in Computation program.
This project has received funding from the European Union’s Horizon 2020 Research and Innovation Programme under the Marie Skłodowska-Curie Grant Agreement No.~750523.

\bibliography{refs3}

\begin{thebibliography}{69}%
\makeatletter
\providecommand \@ifxundefined [1]{%
 \@ifx{#1\undefined}
}%
\providecommand \@ifnum [1]{%
 \ifnum #1\expandafter \@firstoftwo
 \else \expandafter \@secondoftwo
 \fi
}%
\providecommand \@ifx [1]{%
 \ifx #1\expandafter \@firstoftwo
 \else \expandafter \@secondoftwo
 \fi
}%
\providecommand \natexlab [1]{#1}%
\providecommand \enquote  [1]{``#1''}%
\providecommand \bibnamefont  [1]{#1}%
\providecommand \bibfnamefont [1]{#1}%
\providecommand \citenamefont [1]{#1}%
\providecommand \href@noop [0]{\@secondoftwo}%
\providecommand \href [0]{\begingroup \@sanitize@url \@href}%
\providecommand \@href[1]{\@@startlink{#1}\@@href}%
\providecommand \@@href[1]{\endgroup#1\@@endlink}%
\providecommand \@sanitize@url [0]{\catcode `\\12\catcode `\$12\catcode
  `\&12\catcode `\#12\catcode `\^12\catcode `\_12\catcode `\%12\relax}%
\providecommand \@@startlink[1]{}%
\providecommand \@@endlink[0]{}%
\providecommand \url  [0]{\begingroup\@sanitize@url \@url }%
\providecommand \@url [1]{\endgroup\@href {#1}{\urlprefix }}%
\providecommand \urlprefix  [0]{URL }%
\providecommand \Eprint [0]{\href }%
\providecommand \doibase [0]{http://dx.doi.org/}%
\providecommand \selectlanguage [0]{\@gobble}%
\providecommand \bibinfo  [0]{\@secondoftwo}%
\providecommand \bibfield  [0]{\@secondoftwo}%
\providecommand \translation [1]{[#1]}%
\providecommand \BibitemOpen [0]{}%
\providecommand \bibitemStop [0]{}%
\providecommand \bibitemNoStop [0]{.\EOS\space}%
\providecommand \EOS [0]{\spacefactor3000\relax}%
\providecommand \BibitemShut  [1]{\csname bibitem#1\endcsname}%
\let\auto@bib@innerbib\@empty
\bibitem [{\citenamefont {Bell}(1966)}]{bell2}%
  \BibitemOpen
  \bibfield  {author} {\bibinfo {author} {\bibfnamefont {J.~S.}\ \bibnamefont
  {Bell}},\ }\href@noop {} {\bibfield  {journal} {\bibinfo  {journal} {Reviews
  of Modern Physics}\ }\textbf {\bibinfo {volume} {38}},\ \bibinfo {pages}
  {447} (\bibinfo {year} {1966})}\BibitemShut {NoStop}%
\bibitem [{\citenamefont {Kochen}\ and\ \citenamefont {Specker}(1975)}]{ks}%
  \BibitemOpen
  \bibfield  {author} {\bibinfo {author} {\bibfnamefont {S.}~\bibnamefont
  {Kochen}}\ and\ \bibinfo {author} {\bibfnamefont {E.~P.}\ \bibnamefont
  {Specker}},\ }in\ \href@noop {} {\emph {\bibinfo {booktitle} {The
  Logico-Algebraic Approach to Quantum Mechanics}}}\ (\bibinfo  {publisher}
  {Springer},\ \bibinfo {year} {1975})\ pp.\ \bibinfo {pages}
  {263--276}\BibitemShut {NoStop}%
\bibitem [{\citenamefont {Anders}\ and\ \citenamefont
  {Browne}(2009)}]{andersbrowne}%
  \BibitemOpen
  \bibfield  {author} {\bibinfo {author} {\bibfnamefont {J.}~\bibnamefont
  {Anders}}\ and\ \bibinfo {author} {\bibfnamefont {D.~E.}\ \bibnamefont
  {Browne}},\ }\href@noop {} {\bibfield  {journal} {\bibinfo  {journal}
  {Physical Review Letters}\ }\textbf {\bibinfo {volume} {102}},\ \bibinfo
  {pages} {050502} (\bibinfo {year} {2009})}\BibitemShut {NoStop}%
\bibitem [{\citenamefont {Raussendorf}(2013)}]{raussendorfcontextuality}%
  \BibitemOpen
  \bibfield  {author} {\bibinfo {author} {\bibfnamefont {R.}~\bibnamefont
  {Raussendorf}},\ }\href@noop {} {\bibfield  {journal} {\bibinfo  {journal}
  {Physical Review A}\ }\textbf {\bibinfo {volume} {88}},\ \bibinfo {pages}
  {022322} (\bibinfo {year} {2013})}\BibitemShut {NoStop}%
\bibitem [{\citenamefont {Abramsky}\ \emph {et~al.}(2017)\citenamefont
  {Abramsky}, \citenamefont {Barbosa},\ and\ \citenamefont
  {Mansfield}}]{abmprl}%
  \BibitemOpen
  \bibfield  {author} {\bibinfo {author} {\bibfnamefont {S.}~\bibnamefont
  {Abramsky}}, \bibinfo {author} {\bibfnamefont {R.~S.}\ \bibnamefont
  {Barbosa}}, \ and\ \bibinfo {author} {\bibfnamefont {S.}~\bibnamefont
  {Mansfield}},\ }\href@noop {} {\bibfield  {journal} {\bibinfo  {journal}
  {Phys. Rev. Lett.}\ }\textbf {\bibinfo {volume} {119}},\ \bibinfo {pages}
  {050504} (\bibinfo {year} {2017})}\BibitemShut {NoStop}%
\bibitem [{\citenamefont {Karanjai}\ \emph {et~al.}(2018)\citenamefont
  {Karanjai}, \citenamefont {Wallman},\ and\ \citenamefont
  {Bartlett}}]{karanjai2018contextuality}%
  \BibitemOpen
  \bibfield  {author} {\bibinfo {author} {\bibfnamefont {A.}~\bibnamefont
  {Karanjai}}, \bibinfo {author} {\bibfnamefont {J.~J.}\ \bibnamefont
  {Wallman}}, \ and\ \bibinfo {author} {\bibfnamefont {S.~D.}\ \bibnamefont
  {Bartlett}},\ }\href@noop {} {\bibfield  {journal} {\bibinfo  {journal}
  {arXiv preprint arXiv:1802.07744}\ } (\bibinfo {year} {2018})}\BibitemShut
  {NoStop}%
\bibitem [{\citenamefont {Howard}\ \emph {et~al.}(2014)\citenamefont {Howard},
  \citenamefont {Wallman}, \citenamefont {Veitch},\ and\ \citenamefont
  {Emerson}}]{howardmagic}%
  \BibitemOpen
  \bibfield  {author} {\bibinfo {author} {\bibfnamefont {M.}~\bibnamefont
  {Howard}}, \bibinfo {author} {\bibfnamefont {J.}~\bibnamefont {Wallman}},
  \bibinfo {author} {\bibfnamefont {V.}~\bibnamefont {Veitch}}, \ and\ \bibinfo
  {author} {\bibfnamefont {J.}~\bibnamefont {Emerson}},\ }\href@noop {}
  {\bibfield  {journal} {\bibinfo  {journal} {Nature}\ }\textbf {\bibinfo
  {volume} {510}},\ \bibinfo {pages} {351} (\bibinfo {year}
  {2014})}\BibitemShut {NoStop}%
\bibitem [{\citenamefont {Delfosse}\ \emph {et~al.}(2015)\citenamefont
  {Delfosse}, \citenamefont {Guerin}, \citenamefont {Bian},\ and\ \citenamefont
  {Raussendorf}}]{delfosse2015wigner}%
  \BibitemOpen
  \bibfield  {author} {\bibinfo {author} {\bibfnamefont {N.}~\bibnamefont
  {Delfosse}}, \bibinfo {author} {\bibfnamefont {P.~A.}\ \bibnamefont
  {Guerin}}, \bibinfo {author} {\bibfnamefont {J.}~\bibnamefont {Bian}}, \ and\
  \bibinfo {author} {\bibfnamefont {R.}~\bibnamefont {Raussendorf}},\
  }\href@noop {} {\bibfield  {journal} {\bibinfo  {journal} {Physical Review
  X}\ }\textbf {\bibinfo {volume} {5}},\ \bibinfo {pages} {021003} (\bibinfo
  {year} {2015})}\BibitemShut {NoStop}%
\bibitem [{\citenamefont {Pashayan}\ \emph {et~al.}(2015)\citenamefont
  {Pashayan}, \citenamefont {Wallman},\ and\ \citenamefont
  {Bartlett}}]{pashayan2015estimating}%
  \BibitemOpen
  \bibfield  {author} {\bibinfo {author} {\bibfnamefont {H.}~\bibnamefont
  {Pashayan}}, \bibinfo {author} {\bibfnamefont {J.~J.}\ \bibnamefont
  {Wallman}}, \ and\ \bibinfo {author} {\bibfnamefont {S.~D.}\ \bibnamefont
  {Bartlett}},\ }\href@noop {} {\bibfield  {journal} {\bibinfo  {journal}
  {Physical Review Letters}\ }\textbf {\bibinfo {volume} {115}},\ \bibinfo
  {pages} {070501} (\bibinfo {year} {2015})}\BibitemShut {NoStop}%
\bibitem [{\citenamefont {Bermejo-Vega}\ \emph {et~al.}(2017)\citenamefont
  {Bermejo-Vega}, \citenamefont {Delfosse}, \citenamefont {Browne},
  \citenamefont {Okay},\ and\ \citenamefont {Raussendorf}}]{juanito2017}%
  \BibitemOpen
  \bibfield  {author} {\bibinfo {author} {\bibfnamefont {J.}~\bibnamefont
  {Bermejo-Vega}}, \bibinfo {author} {\bibfnamefont {N.}~\bibnamefont
  {Delfosse}}, \bibinfo {author} {\bibfnamefont {D.~E.}\ \bibnamefont
  {Browne}}, \bibinfo {author} {\bibfnamefont {C.}~\bibnamefont {Okay}}, \ and\
  \bibinfo {author} {\bibfnamefont {R.}~\bibnamefont {Raussendorf}},\
  }\href@noop {} {\bibfield  {journal} {\bibinfo  {journal} {Phys. Rev. Lett.}\
  }\textbf {\bibinfo {volume} {119}},\ \bibinfo {pages} {120505} (\bibinfo
  {year} {2017})}\BibitemShut {NoStop}%
\bibitem [{\citenamefont {Catani}\ and\ \citenamefont
  {Browne}(2017)}]{catani2017state}%
  \BibitemOpen
  \bibfield  {author} {\bibinfo {author} {\bibfnamefont {L.}~\bibnamefont
  {Catani}}\ and\ \bibinfo {author} {\bibfnamefont {D.~E.}\ \bibnamefont
  {Browne}},\ }\href@noop {} {\bibfield  {journal} {\bibinfo  {journal} {arXiv
  preprint arXiv:1711.08676}\ } (\bibinfo {year} {2017})}\BibitemShut {NoStop}%
\bibitem [{\citenamefont {Spekkens}\ \emph {et~al.}(2009)\citenamefont
  {Spekkens}, \citenamefont {Buzacott}, \citenamefont {Keehn}, \citenamefont
  {Toner},\ and\ \citenamefont {Pryde}}]{spekkens2009preparation}%
  \BibitemOpen
  \bibfield  {author} {\bibinfo {author} {\bibfnamefont {R.~W.}\ \bibnamefont
  {Spekkens}}, \bibinfo {author} {\bibfnamefont {D.~H.}\ \bibnamefont
  {Buzacott}}, \bibinfo {author} {\bibfnamefont {A.~J.}\ \bibnamefont {Keehn}},
  \bibinfo {author} {\bibfnamefont {B.}~\bibnamefont {Toner}}, \ and\ \bibinfo
  {author} {\bibfnamefont {G.~J.}\ \bibnamefont {Pryde}},\ }\href@noop {}
  {\bibfield  {journal} {\bibinfo  {journal} {Physical review letters}\
  }\textbf {\bibinfo {volume} {102}},\ \bibinfo {pages} {010401} (\bibinfo
  {year} {2009})}\BibitemShut {NoStop}%
\bibitem [{\citenamefont {Chailloux}\ \emph {et~al.}(2016)\citenamefont
  {Chailloux}, \citenamefont {Kerenidis}, \citenamefont {Kundu},\ and\
  \citenamefont {Sikora}}]{chailloux2016optimal}%
  \BibitemOpen
  \bibfield  {author} {\bibinfo {author} {\bibfnamefont {A.}~\bibnamefont
  {Chailloux}}, \bibinfo {author} {\bibfnamefont {I.}~\bibnamefont
  {Kerenidis}}, \bibinfo {author} {\bibfnamefont {S.}~\bibnamefont {Kundu}}, \
  and\ \bibinfo {author} {\bibfnamefont {J.}~\bibnamefont {Sikora}},\
  }\href@noop {} {\bibfield  {journal} {\bibinfo  {journal} {New Journal of
  Physics}\ }\textbf {\bibinfo {volume} {18}},\ \bibinfo {pages} {045003}
  (\bibinfo {year} {2016})}\BibitemShut {NoStop}%
\bibitem [{\citenamefont {Saha}\ \emph {et~al.}(2017)\citenamefont {Saha},
  \citenamefont {Horodecki},\ and\ \citenamefont
  {Paw{\l}owski}}]{saha2017state}%
  \BibitemOpen
  \bibfield  {author} {\bibinfo {author} {\bibfnamefont {D.}~\bibnamefont
  {Saha}}, \bibinfo {author} {\bibfnamefont {P.}~\bibnamefont {Horodecki}}, \
  and\ \bibinfo {author} {\bibfnamefont {M.}~\bibnamefont {Paw{\l}owski}},\
  }\href@noop {} {\bibfield  {journal} {\bibinfo  {journal} {arXiv preprint
  arXiv:1708.04751}\ } (\bibinfo {year} {2017})}\BibitemShut {NoStop}%
\bibitem [{\citenamefont {Schmid}\ and\ \citenamefont
  {Spekkens}(2018)}]{schmid2018contextual}%
  \BibitemOpen
  \bibfield  {author} {\bibinfo {author} {\bibfnamefont {D.}~\bibnamefont
  {Schmid}}\ and\ \bibinfo {author} {\bibfnamefont {R.~W.}\ \bibnamefont
  {Spekkens}},\ }\href@noop {} {\bibfield  {journal} {\bibinfo  {journal}
  {Physical Review X}\ }\textbf {\bibinfo {volume} {8}},\ \bibinfo {pages}
  {011015} (\bibinfo {year} {2018})}\BibitemShut {NoStop}%
\bibitem [{\citenamefont {Saha}\ and\ \citenamefont
  {Chaturvedi}(2018)}]{saha2018preparation}%
  \BibitemOpen
  \bibfield  {author} {\bibinfo {author} {\bibfnamefont {D.}~\bibnamefont
  {Saha}}\ and\ \bibinfo {author} {\bibfnamefont {A.}~\bibnamefont
  {Chaturvedi}},\ }\href@noop {} {\bibfield  {journal} {\bibinfo  {journal}
  {arXiv preprint arXiv:1802.07215}\ } (\bibinfo {year} {2018})}\BibitemShut
  {NoStop}%
\bibitem [{\citenamefont {Ghorai}\ and\ \citenamefont
  {Pan}(2018)}]{ghorai2018optimal}%
  \BibitemOpen
  \bibfield  {author} {\bibinfo {author} {\bibfnamefont {S.}~\bibnamefont
  {Ghorai}}\ and\ \bibinfo {author} {\bibfnamefont {A.}~\bibnamefont {Pan}},\
  }\href@noop {} {\bibfield  {journal} {\bibinfo  {journal} {arXiv preprint
  arXiv:1806.01194}\ } (\bibinfo {year} {2018})}\BibitemShut {NoStop}%
\bibitem [{\citenamefont {Lillystone}\ \emph {et~al.}(2018)\citenamefont
  {Lillystone}, \citenamefont {Wallman},\ and\ \citenamefont
  {Emerson}}]{lillystone2018contextuality}%
  \BibitemOpen
  \bibfield  {author} {\bibinfo {author} {\bibfnamefont {P.}~\bibnamefont
  {Lillystone}}, \bibinfo {author} {\bibfnamefont {J.~J.}\ \bibnamefont
  {Wallman}}, \ and\ \bibinfo {author} {\bibfnamefont {J.}~\bibnamefont
  {Emerson}},\ }\href@noop {} {\bibfield  {journal} {\bibinfo  {journal} {arXiv
  preprint arXiv:1802.06121}\ } (\bibinfo {year} {2018})}\BibitemShut {NoStop}%
\bibitem [{\citenamefont {Knill}\ and\ \citenamefont
  {Laflamme}(1998)}]{knill1998power}%
  \BibitemOpen
  \bibfield  {author} {\bibinfo {author} {\bibfnamefont {E.}~\bibnamefont
  {Knill}}\ and\ \bibinfo {author} {\bibfnamefont {R.}~\bibnamefont
  {Laflamme}},\ }\href@noop {} {\bibfield  {journal} {\bibinfo  {journal}
  {Physical Review Letters}\ }\textbf {\bibinfo {volume} {81}},\ \bibinfo
  {pages} {5672} (\bibinfo {year} {1998})}\BibitemShut {NoStop}%
\bibitem [{\citenamefont {Galv\~{a}o}\ and\ \citenamefont
  {Hardy}(2003)}]{galvao2003substituting}%
  \BibitemOpen
  \bibfield  {author} {\bibinfo {author} {\bibfnamefont {E.~F.}\ \bibnamefont
  {Galv\~{a}o}}\ and\ \bibinfo {author} {\bibfnamefont {L.}~\bibnamefont
  {Hardy}},\ }\href@noop {} {\bibfield  {journal} {\bibinfo  {journal}
  {Physical Review Letters}\ }\textbf {\bibinfo {volume} {90}},\ \bibinfo
  {pages} {087902} (\bibinfo {year} {2003})}\BibitemShut {NoStop}%
\bibitem [{\citenamefont {Dunjko}\ \emph {et~al.}(2016)\citenamefont {Dunjko},
  \citenamefont {Kapourniotis},\ and\ \citenamefont {Kashefi}}]{dkk}%
  \BibitemOpen
  \bibfield  {author} {\bibinfo {author} {\bibfnamefont {V.}~\bibnamefont
  {Dunjko}}, \bibinfo {author} {\bibfnamefont {T.}~\bibnamefont
  {Kapourniotis}}, \ and\ \bibinfo {author} {\bibfnamefont {E.}~\bibnamefont
  {Kashefi}},\ }\href@noop {} {\bibfield  {journal} {\bibinfo  {journal}
  {Quantum Information and Computation}\ }\textbf {\bibinfo {volume} {16}},\
  \bibinfo {pages} {0061} (\bibinfo {year} {2016})}\BibitemShut {NoStop}%
\bibitem [{\citenamefont {Gleason}(1957)}]{gleason1957measures}%
  \BibitemOpen
  \bibfield  {author} {\bibinfo {author} {\bibfnamefont {A.~M.}\ \bibnamefont
  {Gleason}},\ }\href@noop {} {\bibfield  {journal} {\bibinfo  {journal}
  {Journal of mathematics and mechanics}\ }\textbf {\bibinfo {volume} {6}},\
  \bibinfo {pages} {885} (\bibinfo {year} {1957})}\BibitemShut {NoStop}%
\bibitem [{\citenamefont {Spekkens}(2005)}]{spekkenscontextuality}%
  \BibitemOpen
  \bibfield  {author} {\bibinfo {author} {\bibfnamefont {R.~W.}\ \bibnamefont
  {Spekkens}},\ }\href@noop {} {\bibfield  {journal} {\bibinfo  {journal}
  {Physical Review A}\ }\textbf {\bibinfo {volume} {71}},\ \bibinfo {pages}
  {052108} (\bibinfo {year} {2005})}\BibitemShut {NoStop}%
\bibitem [{\citenamefont {Bravyi}\ \emph {et~al.}(2018)\citenamefont {Bravyi},
  \citenamefont {Gosset},\ and\ \citenamefont {Koenig}}]{bravyi2017quantum}%
  \BibitemOpen
  \bibfield  {author} {\bibinfo {author} {\bibfnamefont {S.}~\bibnamefont
  {Bravyi}}, \bibinfo {author} {\bibfnamefont {D.}~\bibnamefont {Gosset}}, \
  and\ \bibinfo {author} {\bibfnamefont {R.}~\bibnamefont {Koenig}},\
  }\href@noop {} {\bibfield  {journal} {\bibinfo  {journal} {Science}\ }\textbf
  {\bibinfo {volume} {362}},\ \bibinfo {pages} {308} (\bibinfo {year}
  {2018})}\BibitemShut {NoStop}%
\bibitem [{\citenamefont {Mermin}(1990{\natexlab{a}})}]{mermin:90a}%
  \BibitemOpen
  \bibfield  {author} {\bibinfo {author} {\bibfnamefont {N.~D.}\ \bibnamefont
  {Mermin}},\ }\href@noop {} {\bibfield  {journal} {\bibinfo  {journal}
  {American Journal of Physics}\ }\textbf {\bibinfo {volume} {58}},\ \bibinfo
  {pages} {731} (\bibinfo {year} {1990}{\natexlab{a}})}\BibitemShut {NoStop}%
\bibitem [{\citenamefont {Mermin}(1990{\natexlab{b}})}]{mermin:90}%
  \BibitemOpen
  \bibfield  {author} {\bibinfo {author} {\bibfnamefont {N.~D.}\ \bibnamefont
  {Mermin}},\ }\href@noop {} {\bibfield  {journal} {\bibinfo  {journal}
  {Physical Review Letters}\ }\textbf {\bibinfo {volume} {65}},\ \bibinfo
  {pages} {3373} (\bibinfo {year} {1990}{\natexlab{b}})}\BibitemShut {NoStop}%
\bibitem [{\citenamefont {Hoban}\ \emph {et~al.}(2011)\citenamefont {Hoban},
  \citenamefont {Campbell}, \citenamefont {Loukopoulos},\ and\ \citenamefont
  {Browne}}]{hoban2011non}%
  \BibitemOpen
  \bibfield  {author} {\bibinfo {author} {\bibfnamefont {M.~J.}\ \bibnamefont
  {Hoban}}, \bibinfo {author} {\bibfnamefont {E.~T.}\ \bibnamefont {Campbell}},
  \bibinfo {author} {\bibfnamefont {K.}~\bibnamefont {Loukopoulos}}, \ and\
  \bibinfo {author} {\bibfnamefont {D.~E.}\ \bibnamefont {Browne}},\
  }\href@noop {} {\bibfield  {journal} {\bibinfo  {journal} {New Journal of
  Physics}\ }\textbf {\bibinfo {volume} {13}},\ \bibinfo {pages} {023014}
  (\bibinfo {year} {2011})}\BibitemShut {NoStop}%
\bibitem [{\citenamefont {Einstein}\ \emph {et~al.}(1935)\citenamefont
  {Einstein}, \citenamefont {Podolsky},\ and\ \citenamefont {Rosen}}]{epr}%
  \BibitemOpen
  \bibfield  {author} {\bibinfo {author} {\bibfnamefont {A.}~\bibnamefont
  {Einstein}}, \bibinfo {author} {\bibfnamefont {B.}~\bibnamefont {Podolsky}},
  \ and\ \bibinfo {author} {\bibfnamefont {N.}~\bibnamefont {Rosen}},\
  }\href@noop {} {\bibfield  {journal} {\bibinfo  {journal} {Physical review}\
  }\textbf {\bibinfo {volume} {47}},\ \bibinfo {pages} {777} (\bibinfo {year}
  {1935})}\BibitemShut {NoStop}%
\bibitem [{\citenamefont {Bell}(1964)}]{bell1}%
  \BibitemOpen
  \bibfield  {author} {\bibinfo {author} {\bibfnamefont {J.~S.}\ \bibnamefont
  {Bell}},\ }\href@noop {} {\bibfield  {journal} {\bibinfo  {journal}
  {Physics}\ }\textbf {\bibinfo {volume} {1}},\ \bibinfo {pages} {195}
  (\bibinfo {year} {1964})}\BibitemShut {NoStop}%
\bibitem [{\citenamefont {Leggett}\ and\ \citenamefont {Garg}(1985)}]{lg}%
  \BibitemOpen
  \bibfield  {author} {\bibinfo {author} {\bibfnamefont {A.~J.}\ \bibnamefont
  {Leggett}}\ and\ \bibinfo {author} {\bibfnamefont {A.}~\bibnamefont {Garg}},\
  }\href@noop {} {\bibfield  {journal} {\bibinfo  {journal} {Physical Review
  Letters}\ }\textbf {\bibinfo {volume} {54}},\ \bibinfo {pages} {857}
  (\bibinfo {year} {1985})}\BibitemShut {NoStop}%
\bibitem [{\citenamefont {Brukner}\ \emph {et~al.}(2004)\citenamefont
  {Brukner}, \citenamefont {Taylor}, \citenamefont {Cheung},\ and\
  \citenamefont {Vedral}}]{entanglementintime}%
  \BibitemOpen
  \bibfield  {author} {\bibinfo {author} {\bibfnamefont {C.}~\bibnamefont
  {Brukner}}, \bibinfo {author} {\bibfnamefont {S.}~\bibnamefont {Taylor}},
  \bibinfo {author} {\bibfnamefont {S.}~\bibnamefont {Cheung}}, \ and\ \bibinfo
  {author} {\bibfnamefont {V.}~\bibnamefont {Vedral}},\ }\href@noop {}
  {\bibfield  {journal} {\bibinfo  {journal} {arXiv preprint quant-ph/0402127}\
  } (\bibinfo {year} {2004})}\BibitemShut {NoStop}%
\bibitem [{\citenamefont {Timpson}\ and\ \citenamefont
  {Maroney}(2013)}]{timpsonmaroney}%
  \BibitemOpen
  \bibfield  {author} {\bibinfo {author} {\bibfnamefont {C.}~\bibnamefont
  {Timpson}}\ and\ \bibinfo {author} {\bibfnamefont {O.}~\bibnamefont
  {Maroney}},\ }\href@noop {} {\bibfield  {journal} {\bibinfo  {journal} {The
  British Journal for the Philosophy of Science}\ } (\bibinfo {year}
  {2013})}\BibitemShut {NoStop}%
\bibitem [{\citenamefont {Allen}\ \emph {et~al.}(2017)\citenamefont {Allen},
  \citenamefont {Maroney},\ and\ \citenamefont {Gogioso}}]{JMAmacro}%
  \BibitemOpen
  \bibfield  {author} {\bibinfo {author} {\bibfnamefont {J.-M.~A.}\
  \bibnamefont {Allen}}, \bibinfo {author} {\bibfnamefont {O.~J.~E.}\
  \bibnamefont {Maroney}}, \ and\ \bibinfo {author} {\bibfnamefont
  {S.}~\bibnamefont {Gogioso}},\ }\href@noop {} {\bibfield  {journal} {\bibinfo
   {journal} {{Quantum}}\ }\textbf {\bibinfo {volume} {1}},\ \bibinfo {pages}
  {13} (\bibinfo {year} {2017})}\BibitemShut {NoStop}%
\bibitem [{\citenamefont {Liang}\ \emph {et~al.}(2011)\citenamefont {Liang},
  \citenamefont {Spekkens},\ and\ \citenamefont {Wiseman}}]{liang:11}%
  \BibitemOpen
  \bibfield  {author} {\bibinfo {author} {\bibfnamefont {Y.-C.}\ \bibnamefont
  {Liang}}, \bibinfo {author} {\bibfnamefont {R.~W.}\ \bibnamefont {Spekkens}},
  \ and\ \bibinfo {author} {\bibfnamefont {H.~M.}\ \bibnamefont {Wiseman}},\
  }\href@noop {} {\bibfield  {journal} {\bibinfo  {journal} {Physics Reports}\
  }\textbf {\bibinfo {volume} {506}},\ \bibinfo {pages} {1} (\bibinfo {year}
  {2011})}\BibitemShut {NoStop}%
\bibitem [{\citenamefont {Pusey}\ \emph {et~al.}(2012)\citenamefont {Pusey},
  \citenamefont {Barrett},\ and\ \citenamefont {Rudolph}}]{pusey:12}%
  \BibitemOpen
  \bibfield  {author} {\bibinfo {author} {\bibfnamefont {M.~F.}\ \bibnamefont
  {Pusey}}, \bibinfo {author} {\bibfnamefont {J.}~\bibnamefont {Barrett}}, \
  and\ \bibinfo {author} {\bibfnamefont {T.}~\bibnamefont {Rudolph}},\
  }\href@noop {} {\bibfield  {journal} {\bibinfo  {journal} {Nature Physics}\
  }\textbf {\bibinfo {volume} {8}},\ \bibinfo {pages} {475} (\bibinfo {year}
  {2012})}\BibitemShut {NoStop}%
\bibitem [{\citenamefont {Colbeck}\ and\ \citenamefont
  {Renner}(2012)}]{colbeck:12}%
  \BibitemOpen
  \bibfield  {author} {\bibinfo {author} {\bibfnamefont {R.}~\bibnamefont
  {Colbeck}}\ and\ \bibinfo {author} {\bibfnamefont {R.}~\bibnamefont
  {Renner}},\ }\href@noop {} {\bibfield  {journal} {\bibinfo  {journal}
  {Physical Review Letters}\ }\textbf {\bibinfo {volume} {108}},\ \bibinfo
  {pages} {150402} (\bibinfo {year} {2012})}\BibitemShut {NoStop}%
\bibitem [{\citenamefont {Colbeck}\ and\ \citenamefont
  {Renner}(2017)}]{colbeck:13}%
  \BibitemOpen
  \bibfield  {author} {\bibinfo {author} {\bibfnamefont {R.}~\bibnamefont
  {Colbeck}}\ and\ \bibinfo {author} {\bibfnamefont {R.}~\bibnamefont
  {Renner}},\ }\href@noop {} {\bibfield  {journal} {\bibinfo  {journal} {New
  Journal of Physics}\ }\textbf {\bibinfo {volume} {19}},\ \bibinfo {pages}
  {013016} (\bibinfo {year} {2017})}\BibitemShut {NoStop}%
\bibitem [{\citenamefont {Hardy}(2013)}]{hardy:13a}%
  \BibitemOpen
  \bibfield  {author} {\bibinfo {author} {\bibfnamefont {L.}~\bibnamefont
  {Hardy}},\ }\href@noop {} {\bibfield  {journal} {\bibinfo  {journal}
  {International Journal of Modern Physics B}\ }\textbf {\bibinfo {volume}
  {27}} (\bibinfo {year} {2013})}\BibitemShut {NoStop}%
\bibitem [{\citenamefont {Montina}(2015)}]{montina:15}%
  \BibitemOpen
  \bibfield  {author} {\bibinfo {author} {\bibfnamefont {A.}~\bibnamefont
  {Montina}},\ }\href@noop {} {\bibfield  {journal} {\bibinfo  {journal}
  {Modern Physics Letters A}\ }\textbf {\bibinfo {volume} {30}} (\bibinfo
  {year} {2015})}\BibitemShut {NoStop}%
\bibitem [{\citenamefont {Mansfield}(2016)}]{mansfield:16}%
  \BibitemOpen
  \bibfield  {author} {\bibinfo {author} {\bibfnamefont {S.}~\bibnamefont
  {Mansfield}},\ }\href@noop {} {\bibfield  {journal} {\bibinfo  {journal}
  {Physical Review A}\ }\textbf {\bibinfo {volume} {94}},\ \bibinfo {pages}
  {042124} (\bibinfo {year} {2016})}\BibitemShut {NoStop}%
\bibitem [{\citenamefont {Allen}(2016)}]{JMAontic}%
  \BibitemOpen
  \bibfield  {author} {\bibinfo {author} {\bibfnamefont {J.-M.~A.}\
  \bibnamefont {Allen}},\ }\href@noop {} {\bibfield  {journal} {\bibinfo
  {journal} {Quantum Studies: Mathematics and Foundations}\ }\textbf {\bibinfo
  {volume} {3}},\ \bibinfo {pages} {161} (\bibinfo {year} {2016})}\BibitemShut
  {NoStop}%
\bibitem [{\citenamefont {Abramsky}\ and\ \citenamefont
  {Brandenburger}(2011)}]{abramskybrandenburger}%
  \BibitemOpen
  \bibfield  {author} {\bibinfo {author} {\bibfnamefont {S.}~\bibnamefont
  {Abramsky}}\ and\ \bibinfo {author} {\bibfnamefont {A.}~\bibnamefont
  {Brandenburger}},\ }\href@noop {} {\bibfield  {journal} {\bibinfo  {journal}
  {New Journal of Physics}\ }\textbf {\bibinfo {volume} {13}},\ \bibinfo
  {pages} {113036} (\bibinfo {year} {2011})}\BibitemShut {NoStop}%
\bibitem [{\citenamefont {Raussendorf}\ and\ \citenamefont
  {Briegel}(2001)}]{raussendorfbriegel}%
  \BibitemOpen
  \bibfield  {author} {\bibinfo {author} {\bibfnamefont {R.}~\bibnamefont
  {Raussendorf}}\ and\ \bibinfo {author} {\bibfnamefont {H.~J.}\ \bibnamefont
  {Briegel}},\ }\href@noop {} {\bibfield  {journal} {\bibinfo  {journal}
  {Physical Review Letters}\ }\textbf {\bibinfo {volume} {86}},\ \bibinfo
  {pages} {5188} (\bibinfo {year} {2001})}\BibitemShut {NoStop}%
\bibitem [{\citenamefont {Damm}(1990)}]{damm1990problems}%
  \BibitemOpen
  \bibfield  {author} {\bibinfo {author} {\bibfnamefont {C.}~\bibnamefont
  {Damm}},\ }\href@noop {} {\bibfield  {journal} {\bibinfo  {journal}
  {Information Processing Letters}\ }\textbf {\bibinfo {volume} {36}},\
  \bibinfo {pages} {247} (\bibinfo {year} {1990})}\BibitemShut {NoStop}%
\bibitem [{\citenamefont {Aaronson}\ and\ \citenamefont
  {Gottesman}(2004)}]{aaronson2004improved}%
  \BibitemOpen
  \bibfield  {author} {\bibinfo {author} {\bibfnamefont {S.}~\bibnamefont
  {Aaronson}}\ and\ \bibinfo {author} {\bibfnamefont {D.}~\bibnamefont
  {Gottesman}},\ }\href@noop {} {\bibfield  {journal} {\bibinfo  {journal}
  {Physical Review A}\ }\textbf {\bibinfo {volume} {70}},\ \bibinfo {pages}
  {052328} (\bibinfo {year} {2004})}\BibitemShut {NoStop}%
\bibitem [{\citenamefont {Bremner}\ \emph {et~al.}(2010)\citenamefont
  {Bremner}, \citenamefont {Jozsa},\ and\ \citenamefont
  {Shepherd}}]{bremner2010classical}%
  \BibitemOpen
  \bibfield  {author} {\bibinfo {author} {\bibfnamefont {M.~J.}\ \bibnamefont
  {Bremner}}, \bibinfo {author} {\bibfnamefont {R.}~\bibnamefont {Jozsa}}, \
  and\ \bibinfo {author} {\bibfnamefont {D.~J.}\ \bibnamefont {Shepherd}},\
  }in\ \href@noop {} {\emph {\bibinfo {booktitle} {Proceedings of the Royal
  Society of London A: Mathematical, Physical and Engineering Sciences}}}\
  (\bibinfo {organization} {The Royal Society},\ \bibinfo {year} {2010})\ p.\
  \bibinfo {pages} {rspa20100301}\BibitemShut {NoStop}%
\bibitem [{\citenamefont {Greenberger}\ \emph {et~al.}(1989)\citenamefont
  {Greenberger}, \citenamefont {Horne},\ and\ \citenamefont
  {Zeilinger}}]{ghz:89}%
  \BibitemOpen
  \bibfield  {author} {\bibinfo {author} {\bibfnamefont {D.~M.}\ \bibnamefont
  {Greenberger}}, \bibinfo {author} {\bibfnamefont {M.~A.}\ \bibnamefont
  {Horne}}, \ and\ \bibinfo {author} {\bibfnamefont {A.}~\bibnamefont
  {Zeilinger}},\ }in\ \href@noop {} {\emph {\bibinfo {booktitle} {Bell’s
  theorem, quantum theory and conceptions of the universe}}}\ (\bibinfo
  {publisher} {Springer},\ \bibinfo {year} {1989})\ pp.\ \bibinfo {pages}
  {69--72}\BibitemShut {NoStop}%
\bibitem [{\citenamefont {Greenberger}\ \emph {et~al.}(1990)\citenamefont
  {Greenberger}, \citenamefont {Horne}, \citenamefont {Shimony},\ and\
  \citenamefont {Zeilinger}}]{ghsz:90}%
  \BibitemOpen
  \bibfield  {author} {\bibinfo {author} {\bibfnamefont {D.~M.}\ \bibnamefont
  {Greenberger}}, \bibinfo {author} {\bibfnamefont {M.~A.}\ \bibnamefont
  {Horne}}, \bibinfo {author} {\bibfnamefont {A.}~\bibnamefont {Shimony}}, \
  and\ \bibinfo {author} {\bibfnamefont {A.}~\bibnamefont {Zeilinger}},\
  }\href@noop {} {\bibfield  {journal} {\bibinfo  {journal} {American Journal
  of Physics}\ }\textbf {\bibinfo {volume} {58}},\ \bibinfo {pages} {1131}
  (\bibinfo {year} {1990})}\BibitemShut {NoStop}%
\bibitem [{\citenamefont {Abramsky}\ \emph {et~al.}(2015)\citenamefont
  {Abramsky}, \citenamefont {Barbosa}, \citenamefont {Kishida}, \citenamefont
  {Lal},\ and\ \citenamefont {Mansfield}}]{ccp}%
  \BibitemOpen
  \bibfield  {author} {\bibinfo {author} {\bibfnamefont {S.}~\bibnamefont
  {Abramsky}}, \bibinfo {author} {\bibfnamefont {R.~S.}\ \bibnamefont
  {Barbosa}}, \bibinfo {author} {\bibfnamefont {K.}~\bibnamefont {Kishida}},
  \bibinfo {author} {\bibfnamefont {R.}~\bibnamefont {Lal}}, \ and\ \bibinfo
  {author} {\bibfnamefont {S.}~\bibnamefont {Mansfield}},\ }in\ \href@noop {}
  {\emph {\bibinfo {booktitle} {24th EACSL Annual Conference on Computer
  Science Logic (CSL 2015)}}},\ \bibinfo {series} {Leibniz International
  Proceedings in Informatics (LIPIcs)}, Vol.~\bibinfo {volume} {41},\ \bibinfo
  {editor} {edited by\ \bibinfo {editor} {\bibfnamefont {S.}~\bibnamefont
  {Kreutzer}}}\ (\bibinfo {year} {2015})\ pp.\ \bibinfo {pages}
  {211--228}\BibitemShut {NoStop}%
\bibitem [{\citenamefont {Grudka}\ \emph {et~al.}(2014)\citenamefont {Grudka},
  \citenamefont {Horodecki}, \citenamefont {Horodecki}, \citenamefont
  {Horodecki}, \citenamefont {Horodecki}, \citenamefont {Joshi}, \citenamefont
  {K{\l}obus},\ and\ \citenamefont {W{\'o}jcik}}]{grudka2014quantifying}%
  \BibitemOpen
  \bibfield  {author} {\bibinfo {author} {\bibfnamefont {A.}~\bibnamefont
  {Grudka}}, \bibinfo {author} {\bibfnamefont {K.}~\bibnamefont {Horodecki}},
  \bibinfo {author} {\bibfnamefont {M.}~\bibnamefont {Horodecki}}, \bibinfo
  {author} {\bibfnamefont {P.}~\bibnamefont {Horodecki}}, \bibinfo {author}
  {\bibfnamefont {R.}~\bibnamefont {Horodecki}}, \bibinfo {author}
  {\bibfnamefont {P.}~\bibnamefont {Joshi}}, \bibinfo {author} {\bibfnamefont
  {W.}~\bibnamefont {K{\l}obus}}, \ and\ \bibinfo {author} {\bibfnamefont
  {A.}~\bibnamefont {W{\'o}jcik}},\ }\href@noop {} {\bibfield  {journal}
  {\bibinfo  {journal} {Physical Review Letters}\ }\textbf {\bibinfo {volume}
  {112}},\ \bibinfo {pages} {120401} (\bibinfo {year} {2014})}\BibitemShut
  {NoStop}%
\bibitem [{\citenamefont {Horodecki}\ \emph {et~al.}(2015)\citenamefont
  {Horodecki}, \citenamefont {Grudka}, \citenamefont {Joshi}, \citenamefont
  {K{\l}obus},\ and\ \citenamefont {{\L}odyga}}]{horodecki2015axiomatic}%
  \BibitemOpen
  \bibfield  {author} {\bibinfo {author} {\bibfnamefont {K.}~\bibnamefont
  {Horodecki}}, \bibinfo {author} {\bibfnamefont {A.}~\bibnamefont {Grudka}},
  \bibinfo {author} {\bibfnamefont {P.}~\bibnamefont {Joshi}}, \bibinfo
  {author} {\bibfnamefont {W.}~\bibnamefont {K{\l}obus}}, \ and\ \bibinfo
  {author} {\bibfnamefont {J.}~\bibnamefont {{\L}odyga}},\ }\href@noop {}
  {\bibfield  {journal} {\bibinfo  {journal} {Physical Review A}\ }\textbf
  {\bibinfo {volume} {92}},\ \bibinfo {pages} {032104} (\bibinfo {year}
  {2015})}\BibitemShut {NoStop}%
\bibitem [{\citenamefont {Amaral}\ \emph {et~al.}(2017)\citenamefont {Amaral},
  \citenamefont {Cabello}, \citenamefont {Cunha},\ and\ \citenamefont
  {Aolita}}]{amaral2017noncontextual}%
  \BibitemOpen
  \bibfield  {author} {\bibinfo {author} {\bibfnamefont {B.}~\bibnamefont
  {Amaral}}, \bibinfo {author} {\bibfnamefont {A.}~\bibnamefont {Cabello}},
  \bibinfo {author} {\bibfnamefont {M.~T.}\ \bibnamefont {Cunha}}, \ and\
  \bibinfo {author} {\bibfnamefont {L.}~\bibnamefont {Aolita}},\ }\href@noop {}
  {\bibfield  {journal} {\bibinfo  {journal} {arXiv preprint arXiv:1705.07911}\
  } (\bibinfo {year} {2017})}\BibitemShut {NoStop}%
\bibitem [{\citenamefont {Duarte}\ and\ \citenamefont
  {Amaral}(2017)}]{duarte2017resource}%
  \BibitemOpen
  \bibfield  {author} {\bibinfo {author} {\bibfnamefont {C.}~\bibnamefont
  {Duarte}}\ and\ \bibinfo {author} {\bibfnamefont {B.}~\bibnamefont
  {Amaral}},\ }\href@noop {} {\bibfield  {journal} {\bibinfo  {journal} {arXiv
  preprint arXiv:1711.10465}\ } (\bibinfo {year} {2017})}\BibitemShut {NoStop}%
\bibitem [{\citenamefont {Henaut}\ \emph {et~al.}(2018)\citenamefont {Henaut},
  \citenamefont {Catani}, \citenamefont {Browne}, \citenamefont {Mansfield},\
  and\ \citenamefont {Pappa}}]{henaut2018tsirelson}%
  \BibitemOpen
  \bibfield  {author} {\bibinfo {author} {\bibfnamefont {L.}~\bibnamefont
  {Henaut}}, \bibinfo {author} {\bibfnamefont {L.}~\bibnamefont {Catani}},
  \bibinfo {author} {\bibfnamefont {D.~E.}\ \bibnamefont {Browne}}, \bibinfo
  {author} {\bibfnamefont {S.}~\bibnamefont {Mansfield}}, \ and\ \bibinfo
  {author} {\bibfnamefont {A.}~\bibnamefont {Pappa}},\ }\href@noop {}
  {\bibfield  {journal} {\bibinfo  {journal} {arXiv preprint arXiv:1806.05624}\
  } (\bibinfo {year} {2018})}\BibitemShut {NoStop}%
\bibitem [{\citenamefont {Barz}\ \emph {et~al.}(2016)\citenamefont {Barz},
  \citenamefont {Dunjko}, \citenamefont {Schlederer}, \citenamefont {Moore},
  \citenamefont {Kashefi},\ and\ \citenamefont {Walmsley}}]{barzenhanced}%
  \BibitemOpen
  \bibfield  {author} {\bibinfo {author} {\bibfnamefont {S.}~\bibnamefont
  {Barz}}, \bibinfo {author} {\bibfnamefont {V.}~\bibnamefont {Dunjko}},
  \bibinfo {author} {\bibfnamefont {F.}~\bibnamefont {Schlederer}}, \bibinfo
  {author} {\bibfnamefont {M.}~\bibnamefont {Moore}}, \bibinfo {author}
  {\bibfnamefont {E.}~\bibnamefont {Kashefi}}, \ and\ \bibinfo {author}
  {\bibfnamefont {I.~A.}\ \bibnamefont {Walmsley}},\ }\href@noop {} {\bibfield
  {journal} {\bibinfo  {journal} {Physical Review A}\ }\textbf {\bibinfo
  {volume} {93}},\ \bibinfo {pages} {032339} (\bibinfo {year}
  {2016})}\BibitemShut {NoStop}%
\bibitem [{\citenamefont {Clementi}\ \emph {et~al.}(2017)\citenamefont
  {Clementi}, \citenamefont {Pappa}, \citenamefont {Eckstein}, \citenamefont
  {Walmsley}, \citenamefont {Kashefi},\ and\ \citenamefont {Barz}}]{mod4}%
  \BibitemOpen
  \bibfield  {author} {\bibinfo {author} {\bibfnamefont {M.}~\bibnamefont
  {Clementi}}, \bibinfo {author} {\bibfnamefont {A.}~\bibnamefont {Pappa}},
  \bibinfo {author} {\bibfnamefont {A.}~\bibnamefont {Eckstein}}, \bibinfo
  {author} {\bibfnamefont {I.~A.}\ \bibnamefont {Walmsley}}, \bibinfo {author}
  {\bibfnamefont {E.}~\bibnamefont {Kashefi}}, \ and\ \bibinfo {author}
  {\bibfnamefont {S.}~\bibnamefont {Barz}},\ }\href@noop {} {\bibfield
  {journal} {\bibinfo  {journal} {arXiv preprint arXiv:1708.06144}\ } (\bibinfo
  {year} {2017})}\BibitemShut {NoStop}%
\bibitem [{\citenamefont {Galv\~{a}o}(2005)}]{galvao2005discrete}%
  \BibitemOpen
  \bibfield  {author} {\bibinfo {author} {\bibfnamefont {E.~F.}\ \bibnamefont
  {Galv\~{a}o}},\ }\href@noop {} {\bibfield  {journal} {\bibinfo  {journal}
  {Physical Review A}\ }\textbf {\bibinfo {volume} {71}},\ \bibinfo {pages}
  {042302} (\bibinfo {year} {2005})}\BibitemShut {NoStop}%
\bibitem [{\citenamefont {Veitch}\ \emph {et~al.}(2012)\citenamefont {Veitch},
  \citenamefont {Ferrie}, \citenamefont {Gross},\ and\ \citenamefont
  {Emerson}}]{veitch2012negative}%
  \BibitemOpen
  \bibfield  {author} {\bibinfo {author} {\bibfnamefont {V.}~\bibnamefont
  {Veitch}}, \bibinfo {author} {\bibfnamefont {C.}~\bibnamefont {Ferrie}},
  \bibinfo {author} {\bibfnamefont {D.}~\bibnamefont {Gross}}, \ and\ \bibinfo
  {author} {\bibfnamefont {J.}~\bibnamefont {Emerson}},\ }\href@noop {}
  {\bibfield  {journal} {\bibinfo  {journal} {New Journal of Physics}\ }\textbf
  {\bibinfo {volume} {14}},\ \bibinfo {pages} {113011} (\bibinfo {year}
  {2012})}\BibitemShut {NoStop}%
\bibitem [{\citenamefont {Spekkens}(2008)}]{spekkens2008negativity}%
  \BibitemOpen
  \bibfield  {author} {\bibinfo {author} {\bibfnamefont {R.~W.}\ \bibnamefont
  {Spekkens}},\ }\href@noop {} {\bibfield  {journal} {\bibinfo  {journal}
  {Physical Review Letters}\ }\textbf {\bibinfo {volume} {101}},\ \bibinfo
  {pages} {020401} (\bibinfo {year} {2008})}\BibitemShut {NoStop}%
\bibitem [{\citenamefont {Ac{\'\i}n}\ \emph {et~al.}(2015)\citenamefont
  {Ac{\'\i}n}, \citenamefont {Fritz}, \citenamefont {Leverrier},\ and\
  \citenamefont {Sainz}}]{afls}%
  \BibitemOpen
  \bibfield  {author} {\bibinfo {author} {\bibfnamefont {A.}~\bibnamefont
  {Ac{\'\i}n}}, \bibinfo {author} {\bibfnamefont {T.}~\bibnamefont {Fritz}},
  \bibinfo {author} {\bibfnamefont {A.}~\bibnamefont {Leverrier}}, \ and\
  \bibinfo {author} {\bibfnamefont {A.~B.}\ \bibnamefont {Sainz}},\ }\href@noop
  {} {\bibfield  {journal} {\bibinfo  {journal} {Communications in Mathematical
  Physics}\ }\textbf {\bibinfo {volume} {334}},\ \bibinfo {pages} {533}
  (\bibinfo {year} {2015})}\BibitemShut {NoStop}%
\bibitem [{\citenamefont {Cabello}\ \emph {et~al.}(2014)\citenamefont
  {Cabello}, \citenamefont {Severini},\ and\ \citenamefont {Winter}}]{csw}%
  \BibitemOpen
  \bibfield  {author} {\bibinfo {author} {\bibfnamefont {A.}~\bibnamefont
  {Cabello}}, \bibinfo {author} {\bibfnamefont {S.}~\bibnamefont {Severini}}, \
  and\ \bibinfo {author} {\bibfnamefont {A.}~\bibnamefont {Winter}},\
  }\href@noop {} {\bibfield  {journal} {\bibinfo  {journal} {Physical Review
  Letters}\ }\textbf {\bibinfo {volume} {112}},\ \bibinfo {pages} {040401}
  (\bibinfo {year} {2014})}\BibitemShut {NoStop}%
\bibitem [{\citenamefont {Dzhafarov}\ \emph {et~al.}(2015)\citenamefont
  {Dzhafarov}, \citenamefont {Kujala},\ and\ \citenamefont {Cervantes}}]{cbd}%
  \BibitemOpen
  \bibfield  {author} {\bibinfo {author} {\bibfnamefont {E.~N.}\ \bibnamefont
  {Dzhafarov}}, \bibinfo {author} {\bibfnamefont {J.~V.}\ \bibnamefont
  {Kujala}}, \ and\ \bibinfo {author} {\bibfnamefont {V.~H.}\ \bibnamefont
  {Cervantes}},\ }in\ \href@noop {} {\emph {\bibinfo {booktitle} {International
  Symposium on Quantum Interaction}}}\ (\bibinfo {organization} {Springer},\
  \bibinfo {year} {2015})\ pp.\ \bibinfo {pages} {12--23}\BibitemShut {NoStop}%
\bibitem [{\citenamefont {de~Silva}(2017)}]{horse1}%
  \BibitemOpen
  \bibfield  {author} {\bibinfo {author} {\bibfnamefont {N.}~\bibnamefont
  {de~Silva}},\ }\href@noop {} {\bibfield  {journal} {\bibinfo  {journal}
  {Physical Review A}\ }\textbf {\bibinfo {volume} {95}},\ \bibinfo {pages}
  {032108} (\bibinfo {year} {2017})}\BibitemShut {NoStop}%
\bibitem [{\citenamefont {Kujala}\ \emph {et~al.}(2015)\citenamefont {Kujala},
  \citenamefont {Dzhafarov},\ and\ \citenamefont {Larsson}}]{cbdmacro}%
  \BibitemOpen
  \bibfield  {author} {\bibinfo {author} {\bibfnamefont {J.~V.}\ \bibnamefont
  {Kujala}}, \bibinfo {author} {\bibfnamefont {E.~N.}\ \bibnamefont
  {Dzhafarov}}, \ and\ \bibinfo {author} {\bibfnamefont {J.-{\AA}.}\
  \bibnamefont {Larsson}},\ }\href@noop {} {\bibfield  {journal} {\bibinfo
  {journal} {Physical Review Letters}\ }\textbf {\bibinfo {volume} {115}},\
  \bibinfo {pages} {150401} (\bibinfo {year} {2015})}\BibitemShut {NoStop}%
\bibitem [{\citenamefont {Wester}(2018)}]{wester}%
  \BibitemOpen
  \bibfield  {author} {\bibinfo {author} {\bibfnamefont {L.}~\bibnamefont
  {Wester}},\ }in\ \href@noop {} {\emph {\bibinfo {booktitle} {{\rm Proceedings
  14th International Conference on} Quantum Physics and Logic, {\rm Nijmegen,
  The Netherlands, 3-7 July 2017}}}},\ \bibinfo {series} {Electronic
  Proceedings in Theoretical Computer Science}, Vol.\ \bibinfo {volume} {266},\
  \bibinfo {editor} {edited by\ \bibinfo {editor} {\bibfnamefont
  {B.}~\bibnamefont {Coecke}}\ and\ \bibinfo {editor} {\bibfnamefont
  {A.}~\bibnamefont {Kissinger}}}\ (\bibinfo  {publisher} {Open Publishing
  Association},\ \bibinfo {year} {2018})\ pp.\ \bibinfo {pages}
  {1--22}\BibitemShut {NoStop}%
\bibitem [{\citenamefont {Kunjwal}(2017)}]{kunjwal2017beyond}%
  \BibitemOpen
  \bibfield  {author} {\bibinfo {author} {\bibfnamefont {R.}~\bibnamefont
  {Kunjwal}},\ }\href@noop {} {\bibfield  {journal} {\bibinfo  {journal} {arXiv
  preprint arXiv:1709.01098}\ } (\bibinfo {year} {2017})}\BibitemShut {NoStop}%
\bibitem [{\citenamefont {Amselem}\ \emph {et~al.}(2012)\citenamefont
  {Amselem}, \citenamefont {Danielsen}, \citenamefont {L{\'o}pez-Tarrida},
  \citenamefont {Portillo}, \citenamefont {Bourennane},\ and\ \citenamefont
  {Cabello}}]{nccontent}%
  \BibitemOpen
  \bibfield  {author} {\bibinfo {author} {\bibfnamefont {E.}~\bibnamefont
  {Amselem}}, \bibinfo {author} {\bibfnamefont {L.~E.}\ \bibnamefont
  {Danielsen}}, \bibinfo {author} {\bibfnamefont {A.~J.}\ \bibnamefont
  {L{\'o}pez-Tarrida}}, \bibinfo {author} {\bibfnamefont {J.~R.}\ \bibnamefont
  {Portillo}}, \bibinfo {author} {\bibfnamefont {M.}~\bibnamefont
  {Bourennane}}, \ and\ \bibinfo {author} {\bibfnamefont {A.}~\bibnamefont
  {Cabello}},\ }\href@noop {} {\bibfield  {journal} {\bibinfo  {journal}
  {Physical Review Letters}\ }\textbf {\bibinfo {volume} {108}},\ \bibinfo
  {pages} {200405} (\bibinfo {year} {2012})}\BibitemShut {NoStop}%
\bibitem [{\citenamefont {Mansfield}(2013)}]{thesis}%
  \BibitemOpen
  \bibfield  {author} {\bibinfo {author} {\bibfnamefont {S.}~\bibnamefont
  {Mansfield}},\ }\href@noop {} {\enquote {\bibinfo {title} {The mathematical
  structure of non-locality \& contextuality},}\ }\bibinfo {howpublished}
  {D.Phil. thesis, Oxford University} (\bibinfo {year} {2013})\BibitemShut
  {NoStop}%
\bibitem [{\citenamefont {Abramsky}\ \emph {et~al.}(2016)\citenamefont
  {Abramsky}, \citenamefont {Barbosa},\ and\ \citenamefont
  {Mansfield}}]{abm-qpl16}%
  \BibitemOpen
  \bibfield  {author} {\bibinfo {author} {\bibfnamefont {S.}~\bibnamefont
  {Abramsky}}, \bibinfo {author} {\bibfnamefont {R.~S.}\ \bibnamefont
  {Barbosa}}, \ and\ \bibinfo {author} {\bibfnamefont {S.}~\bibnamefont
  {Mansfield}},\ }\href@noop {} {\bibfield  {journal} {\bibinfo  {journal}
  {Informal Proceedings of Quantum Physics \& Logic}\ } (\bibinfo {year}
  {2016})}\BibitemShut {NoStop}%
\end{thebibliography}%
\bibliographystyle{apsrev4-1}


\section{Appendix}
\subsection{Comparison with Spekkens' contextuality}

Within the ontological models framework as used by Spekkens~\cite{spekkenscontextuality},
in addition to the basic ingredients of ontological models that we have set out in the main text, a number of further features are imposed
or implicitly assumed
motivated by the intended interpretation of ontological models. Here, we have chosen to set out a more minimal definition of what we intend to mean by ontological models, and to explicitly state any such additional assumptions as they become relevant, preferring to see these as crucial to considerations and definitions of contextuality.

Two such features are that sequential composition should be respected at the ontological level,
\[
f_{U_t\cdots U_1} = f_{U_t} \circ \dots \circ f_{U_1} \, ;
\]
and that $f_{U^{(C)}} = f_{U^{(C')}}$ for sequential contexts ($C$) and ($C'$). Since these are the components of our definition of sequential transformation non-contextuality, then non-contextuality in our sense could be thought of as being implicitly baked-in to the more loaded version of the ontological models framework from the outset. From our perspective, however, this would be undesirable since it would fail to pick up on sequential transformation contextuality, a non-classical phenomenon that on the basis of our results we believe to be worthy of consideration.

Spekkens' generalised approach to non-contextuality is that ontological identifications such as $f_{U^{(C)}} = f_{U^{(C')}}$ should be imposed whenever there is an operational equivalence. Here, $U^{(C)}$ and $U^{(C')}$ would be said to be operationally equivalent if, for all choices of preparation and measurement, the outcome statistics for the prepare-transform-measure experiments with transformation $U^{(C)}$ and with transformation $U^{(C')}$ were equal. However, since $U^{(C)}$ designates the transformation $U$ when it appears in the sequence of transformations ($C$), and similarly for $U^{(C')}$, such statistics are operationally inaccessible. Without broadening what it means to be operationally equivalent, our notion of non-contextuality is not captured by the Spekkens approach.

Our approach to non-contextuality, on the other hand, is to assume that operational \emph{compositions} are respected at the ontological level, and that ontological representations are independent of operational context. Non-contextuality in the BKS and Spekkens senses are captured by this perspective as well, where now composition does not refer exclusively to sequential transformation of transformations, but also to composition of compatible measurements into a joint measurement, or composition of transformations or preparations through stochastic mixtures, etc. The presentation of the various notions of non-contextuality in the main text aims to facilitate this perspective, which will be more fully developed in a future article.


\subsection{Commutativity in $\oplus L$-ontologies}

In a $\oplus L$-ontology, transformations are built from $\mathsf{CNOT}$ and $\mathsf{NOT}$ gates. The action of the $\mathsf{CNOT}(i,j)$ gate with control bit $i$ and target bit $j$ on an ontic state $\bm{\lambda} \in \mathbb{Z}_2^s$ is a linear operation
\[
\mathsf{CNOT}(i,j) \, \bm{\lambda} = \left( I \oplus A(j,i) \right) \bm{\lambda} \, ,
\]
where $I$ and $A(j,i)$ are $s \times s$ matrices over $\mathbb{Z}_2$, the former being the identity matrix and the latter the matrix whose only non-zero entry is at position $(j,i)$. For composition of $\mathsf{CNOT}$ gates we have
\begin{align*}
\mathsf{CNOT}(k,l) \circ \mathsf{CNOT}(i,j) = \left( I \oplus A(l,k) \right) \left( I \oplus A(j,i) \right)& \\
= I \oplus A(l,k) \oplus A(j,i) \oplus \delta_{kj} A(l,i)& \, ,
\end{align*}
and similarly
\begin{align*}
\mathsf{CNOT}(i,j) \circ \mathsf{CNOT}(k,l) = \left( I \oplus A(j,i) \right) \left( I \oplus A(l,k) \right)& \\
= I \oplus A(j,i) \oplus A(l,k) \oplus \delta_{il} A(j,k)& \, .
\end{align*}
The gates commute when $k \neq j, i \neq l$; i.e.~the control bit for one gate cannot be the target bit for the other and vice versa.

The other basic building blocks for transformations are $\mathsf{NOT}$ gates. As a linear operation, the action of a $\mathsf{NOT}$ gate on the $i$th bit is simply addition by the vector $\bm{\delta}(i)$ whose only non-zero entry is in the $i$th position,
\[
\mathsf{NOT}(i) (\bm{\lambda}) = \bm{\lambda} \oplus \bm{\delta}(i) \, .
\]
$\mathsf{NOT}$ gates commute amongst themselves, while a $\mathsf{NOT}$ gate commutes with a $\mathsf{CNOT}$ gate if and only if it does not act on the control bit. With $\mathsf{NOT}$ acting on the target bit, it holds that
\begin{align*}
\mathsf{CNOT}(i,j) \circ \mathsf{NOT}(j) (\bm{\lambda}) = \left( I \oplus A(j,i) \right) \left(\bm{\lambda} \oplus \bm{\delta}(j)\right)& \\
= \left(I \oplus A(j,i) \right) \bm{\lambda} \oplus \bm{\delta}(j)& \\
= \mathsf{NOT}(j) \circ \mathsf{CNOT}(i,j) (\bm{\lambda})& \, ,
\end{align*}
where the second equality follows from the fact that $i\neq j$ since there exists a $\mathsf{CNOT}$ between the respective bits. With $\mathsf{NOT}$ acting on the control bit,
\begin{align*}
\mathsf{CNOT}(i,j) \circ \mathsf{NOT}(i) (\bm{\lambda}) = \left( I \oplus A(j,i) \right) \left(\bm{\lambda} \oplus \bm{\delta}(i)\right)& \\
= \left(I \oplus A(j,i) \right) \bm{\lambda} \oplus \bm{\delta}(i) \oplus \bm{\delta}(j)& \, ,
\end{align*}
whereas
\[
\mathsf{NOT}(i) \circ \mathsf{CNOT}(i,j) (\bm{\lambda}) = \left( I \oplus A(j,i) \right) \bm{\lambda} \oplus \bm{\delta}(i) \, .
\]

In \emph{commutative $\oplus L$-ontologies} we will therefore assume that the ontic state space $\mathbb{Z}_2^s$ can be partitioned into control and target bits and that $\mathsf{NOT}$ gates act only on target bits.
This is an obvious sufficient condition for commutativity of all transformations, though weaker conditions that ensure commutativity only for subsets of the possible transformations may be interesting to consider in future work.
As a linear operation, the ontological representation of any transformation $U$ is given by
\begin{equation}\label{eq:ontrep}
f_U(\bm{\lambda}) = \left( I \oplus A_U \right) \bm{\lambda} \oplus \bm{u} \, ,
\end{equation}
where $A_U$ is some $s \times s$ matrix over $\mathbb{Z}_2$ containing only off-diagonal entries, $\bm{u} \in \mathbb{Z}_2^s$ gives the combined action of any $\mathsf{NOT}$ gates, and for composition we have
\begin{align*}
f_{U_t} \circ \cdots \circ f_{U_1}(\bm{\lambda}) &= \left( I \oplus \bigoplus_{i=1}^{t} A_{U_i} \right) \bm{\lambda} \oplus \bigoplus_{i=1}^{t} \bm{u_i} \\
&= \bm{\lambda} \oplus \bigoplus_{i=1}^t A_{U_{i}} \bm{\lambda} \oplus \bigoplus_{i=1}^t \bm{u}_i \, .
\end{align*}

\end{document}